\documentclass{article}

\usepackage{amsfonts,amsmath,amsthm}
\usepackage{paralist}
\usepackage{multirow}
\usepackage{comment}
\usepackage{xspace}
\usepackage[noend,ruled]{algorithm2e}
\usepackage{colortbl}
\usepackage{tikz}
\usetikzlibrary{decorations,arrows,petri,topaths,backgrounds,shapes,positioning,fit,calc,decorations.pathreplacing,patterns,intersections}
\usepackage[numbers,sectionbib]{natbib}
\usepackage{subcaption}

\usepackage{authblk}

\newlength{\additionaltextwidth}
\usepackage[colorlinks,citecolor=green!60!black,linkcolor=red!60!black,pagebackref]{hyperref}
\renewcommand*{\backref}[1]{}
\renewcommand*{\backrefalt}[4]{%
\ifcase #1%
\marginpar{\tiny no cite}
\or
 $\rightarrow$~p.~#2.%
\else
  $\rightarrow$~pp.~#2.%
\fi
}

\newcommand{\mytitle}{Combinatorial Voter Control in Elections}

\title{\mytitle\thanks{LB was supported by the Alexander von Humboldt Foundation. JC was supported by the Studienstiftung des Deutschen Volkes. PF has been supported by the DFG project PAWS (NI 369/10). NT is supported by the DFG Research Training Group ``Methods for Discrete Structures''~(GRK~1408). This work has been partly supported by COST Action IC1205 on Computational Social Choice.}}

\author[1]{Laurent Bulteau}
\author[1]{Jiehua Chen}
\author[2]{Piotr Faliszewski}
\author[1]{Rolf Niedermeier}
\author[1]{Nimrod Talmon}

\affil[1]{Institut f\"ur Softwaretechnik und Theoretische Informatik,
  TU Berlin, Germany, \texttt{l.bulteau@gmail.com, jiehua.chen@tu-berlin.de, nimrodtalmon77@gmail.com, rolf.niedermeier@tu-berlin.de}}

\affil[2]{AGH University of Science and Technology, Krakow, Poland,
  \texttt{faliszew@agh.edu.pl}}

\date{}

\newcommand{\probCCCAVLong}{\textsc{Combinatorial Constructive Control by Adding Voters}\xspace}
\newcommand{\probCCAVLong}{\textsc{Constructive Control by Adding Voters}\xspace}
\newcommand{\probCCAV}{\textsc{CC-AV}\xspace}
\newcommand{\probCCCAV}{\textsc{C-CC-AV}\xspace}
\newcommand{\probCCCAVInstance}{\ensuremath{((\electionC, \electionV), W, d, \combRule, p\in \electionC, k)}}

\newcommand{\probSetCover}{\textsc{Set Cover}\xspace}

\newcommand{\probClique}{\textsc{Clique}\xspace}
\newcommand{\probPVC}{\textsc{Partial Vertex Cover}\xspace}
\newcommand{\probVC}{\textsc{Vertex Cover}\xspace}
\newcommand{\PVC}{\textsc{PVC}\xspace}

\newcommand{\probSAT}{\textsc{(2,2)-3SAT}\xspace}

\newcommand{\np}{{\mathsf{NP}}}
\newcommand{\fpt}{{\mathsf{FPT}}}
\newcommand{\ilpfpt}{{\mathsf{ILP}\textrm{-}\mathsf{FPT}}}
\newcommand{\xp}{{\mathsf{XP}}}
\newcommand{\wone}{{\mathsf{W[1]}}}
\newcommand{\wtwo}{{\mathsf{W[2]}}}

\newcommand{\p}{{\mathsf{P}}}

\newcommand{\pref}{\ensuremath{\succ}}
\newcommand{\calR}{\ensuremath{\mathcal{R}}}
\newcommand{\electionC}{\ensuremath{C}}
\newcommand{\electionV}{\ensuremath{V}}
\newcommand{\electionW}{\ensuremath{W}}

\newcommand{\combRule}{\ensuremath{\kappa}\xspace}

\newcommand{\assignment}{bun\-dling function\xspace}
\newcommand{\assignmentnoxspace}{bun\-dling function}
\newcommand{\assignments}{bun\-dling functions\xspace}
\newcommand{\fullassignment}[1][d]{full-$#1$ bun\-dling function\xspace}
\newcommand{\fullassignments}[1][d]{full-$#1$ bun\-dling functions\xspace}
\newcommand{\nonfullassignment}[1][d]{$#1$-bounded bundling function\xspace}

\newcommand{\anonassignment}{anonymous bun\-dling function\xspace}
\newcommand{\anonassignments}{anonymous bun\-dling functions\xspace}

\newcommand{\swapdist}{swap distance\xspace}

\newcommand{\bundlinggraph}{bundling graph\xspace}
\newcommand{\Bundlinggraph}{Bundling graph\xspace}

\newcommand{\seq}[1]{\ensuremath{\langle #1\rangle}}

\newcommand{\leaderanonymous}{leader-anonymous\xspace}
\newcommand{\followeranonymous}{follower-anonymous\xspace}

\newcommand{\votertype}[1]{$#1$-voter\xspace}
\newcommand{\voterstype}[1]{$#1$-voters\xspace}

\newcommand{\firstvoter}[1]{\ensuremath{1\mathrm{st}(#1)}}
\newcommand{\secondvoter}[1]{\ensuremath{2\mathrm{nd}(#1)}}

\usepackage{cleveref}

\newtheorem{definition}{Definition}

\newtheorem{theorem}{Theorem}
\newtheorem{lemma}{Lemma}
\newtheorem{observation}{Observation}

\crefname{subsection}{Subsection}{Subsections}
\crefname{section}{Section}{Sections}

\crefname{table}{Table}{Tables}
\crefname{figure}{Figure}{Figures}
\crefname{algorithm}{Algorithm}{Algorithms}
\crefname{theorem}{Theorem}{Theorems}
\crefname{definition}{Definition}{Definitions}
\crefname{corollary}{Corollary}{Corollary}
\crefname{proposition}{Proposition}{Propositions}
\crefname{observation}{Observation}{Observations}
\crefname{lemma}{Lemma}{Lemmas}
\crefname{example}{Example}{Examples}
\crefname{reduction}{Reduction}{Reductions}
\crefname{algorithm}{Algorithm}{Algorithms}
\crefname{appendix}{Appendix}{Appendices}

\newcommand{\probDef}[3]{
  \begin{quote}
   #1\\
  \textbf{Input:} #2\\
  \textbf{Question:} #3
  \end{quote}
}

\newcommand{\mytabref}[1]{[\hyperref[#1]{Thm.~\ref*{#1}}]}
\newcommand{\mytabobs}[1]{[\hyperref[#1]{Obs.~\ref*{#1}}]}

\newcommand{\obsCCCAVinXPk}{
    Both \textsc{Plurality}-\probCCCAV and \textsc{Condorcet}-\probCCCAV
    are solvable in $O(n^{k}\cdot n \cdot m \cdot \mathrm{winner})$
    time, where $\mathrm{winner}$ is the complexity of determining
    Plurality/Condorcet winners.
}

\newcommand{\thmwonekanonbthree}{\textsc{Plurality}-\probCCCAV is $\np$-hard and $\wone$-hard when parameterized by the solution size~$k$, 
  even when the maximum
  bundle size~$b$ is three and the \assignment is anonymous.
}

\newcommand{\thmwtwokmtwo}{Both \textsc{Plurality}-\probCCCAV and \textsc{Condorcet}-\probCCCAV parameterized by the solution size~$k$ are $\wtwo$-hard, 
  even for two alternatives.}

\newcommand{\thmfptm}{For anonymous \assignments, 
  both \textsc{Plurality}-\probCCCAV and \textsc{Condorcet}-\probCCCAV pa\-ram\-e\-terized by the number~$m$ of alternatives
  are fixed-parameter tractable.}

\newcommand{\thmnphbtwo}{\textsc{Plurality}-\probCCCAV is $\np$-hard even if the maximum bundle size~$b$ is two.}
  
\newcommand{\thmfullbthreehard}{
  If \combRule is a \fullassignment, then \textsc{Plurality}-\probCCCAV is $\np$-hard even if the maximum bundle size~$b$ is three.  
}

\newcommand{\thmnphdonebfour}{
  \textsc{Plurality}-\probCCCAV is $\np$-hard even for \fullassignments[1] and
  even if the maximum bundle size~$b$ is four.
}

\newcommand{\thmspwonehk}{
  Both \textsc{Plurality}-\probCCCAV and \textsc{Condorcet}-\probCCCAV parameterized by the solution size~$k$
  are $\wone$-hard
  for single-peaked elections,
  even for \fullassignments[1].
}

\newcommand{\lemplusc}{
  Let $I=\probCCCAVInstance$ be a \textsc{Plurality}-\probCCCAV instance such that $(\electionC, \electionV\cup \electionW)$ is single-crossing
  and $\combRule$ is a \fullassignment. 
  Then, the following statements hold:
  \begin{enumerate}[(i)]
    \item\label{lem:p-voters-consecutive} The \voterstype{p} are ordered consecutively along the single-crossing order.
    \item\label{lem:|non-p-voter-bundles|<=2} If $I$ is a yes instance, 
    then there is a subset~$\electionW' \subseteq \electionW$ of size at most $k$
    such that all bundles of voters $w \in \electionW'$ contain only \voterstype{p},
    except at most two bundles which may contain some non-\voterstype{p}.
  \end{enumerate}
}

\newcommand{\lemcondsc}{
  Let $(\electionC, \electionV \cup \combRule(\electionW'))$ be a single-crossing
  election with single-crossing voter order~$\langle x_1, x_2,$
  $\ldots, x_{z}\rangle$ and
  set $X_{\mathrm{median}}:=\{x_{\lceil
    z/2 \rceil}\} \cup \{x_{z/2 + 1} \text{ if } z \text{ is even}\}$,
  where $z=|V|+|\combRule(\electionW')|$.
  Alternative~$p$ is a (unique) Condorcet winner in $(\electionC,
  \electionV \cup \combRule(\electionW'))$ if and only if every voter
  in $X_{\mathrm{median}}$ is a \votertype{p}.
}

\newcommand{\thmPsc}{
  Both \textsc{Plurality}-\probCCCAV and \textsc{Condorcet}-\probCCCAV
  are po\-ly\-nomial-time solvable for the single-crossing case with \fullassignments.
}

\newcommand{\obsbtwofullp}
{  If $\combRule$ is a \fullassignment and 
  the maximum bundle size~$b$ is two, 
  then \textsc{Plurality}-\probCCCAV is polynomial-time solvable. }

\begin{document}

\maketitle
\thispagestyle{plain}
\setcounter{footnote}{0}

\begin{abstract}
  Voter control problems model situations such as an external agent trying to
  affect the result of an election by adding voters, for example by convincing
  some voters to vote who would otherwise not attend the
  election. Traditionally, voters are added one at a time, with the goal of
  making a distinguished alternative win by adding a minimum number of
  voters. In this paper, we initiate the study of combinatorial variants of
  control by adding voters: In our setting, when we choose to add a voter~$v$,
  we also have to add a whole bundle $\kappa(v)$ of voters associated with
  $v$.  We study the computational complexity of this problem for two of the
  most basic voting rules, namely the Plurality rule and the Condorcet rule.
\end{abstract}

\section{Introduction}
We study the computational complexity of control by adding
voters~\cite{BTT92,HHR07}, investigating the case where the sets of voters that
we can add have some combinatorial structure.
The problem of election control by adding voters models situations
where some agent (e.g., a campaign manager for one of the
alternatives) tries to ensure a given alternative's victory by
convincing some undecided voters to vote. Traditionally, in this
problem we are given a description of an election
(that is, a set~$C$ of alternatives and a set~$V$ of voters who decided to vote),
and also a set $W$ of undecided voters
(for each voter in $V \cup W$ we assume to know how this
voter intends to vote which is given by a linear order of the set~$C$;
we might have good approximation of this knowledge
from preelection polls).
Our goal is to ensure that our preferred
alternative $p$ becomes a winner, by convincing as few voters from $W$
to vote as possible (provided that it is at all possible to ensure
$p$'s victory in this way).

Control by adding voters corresponds, for example, to situations where
supporters of a given alternative make direct appeals to other
supporters of the alternative to vote (for example, they may stress the
importance of voting, or help with the voting process by offering
rides to the voting locations, etc.).  Unfortunately, in its
traditional phrasing, control by adding voters does not model
larger-scale attempts at convincing people to vote. For example, a
campaign manager might be interested in airing a TV advertisement that
would motivate supporters of a given alternative to vote (though, of
course, it might also motivate some of this alternative's enemies),
or maybe 
launch viral
campaigns, where friends convince their own friends to vote. 
It is clear that the sets of voters that we can add
should have some sort of a combinatorial structure.
For instance, 
a TV advertisement appeals to a particular group of voters and we can add
all of them at the unit cost of airing the advertisement.
A public speech in a given neighborhood will convince a particular group of
people to vote at a unit cost of organizing the meeting or convincing a
person to vote will ``for free'' also convince her friends to vote.

The goal of our work is to formally define an appropriate computational problem modeling a
combinatorial variant of control by adding voters and to study
its computational complexity. Specifically, we focus on the
Plurality rule and the Condorcet rule,
mainly because 
the Plurality rule is the most widely used rule in
practice, 
and it is one of the few rules for which the
standard variant of control by adding voters is solvable in polynomial
time~\cite{BTT92}, whereas for the Condorcet rule the problem is
polynomial-time solvable for the case of single-peaked
elections~\cite{FHHR11}. For the case of single-peaked
elections, in essence, all our hardness results for the Condorcet rule
directly translate to all Condorcet-consistent voting rules, a large
and important family of voting rules.
We defer the formal details, definitions, and concrete results to
the following sections. Instead, we state the high-level, main 
messages of our work:
\begin{itemize}
\item 
  Many typical
  variants of combinatorial control by adding voters are intractable,
  but there is also a rich landscape of tractable cases.
\item Assuming that voters have single-peaked preferences does not
  lower the complexity of the problem (even though it does so in many
  election problems~\cite{BBHH2010,Con09,FHHR11}). 
  On the contrary, assuming single-crossing preferences does lower the
  complexity of the problem.
\end{itemize}
We believe that our setting of combinatorial control, and---more
generally---combinatorial voting, 
offers a very fertile ground for future research and we intend
the current paper as an initial step.

\medskip
\noindent
\emph{Related Work.}
\citet{BTT92} first studied the concept of election control by
adding/deleting voters or alternatives in a given election. 
They studied the constructive variant of the problem, where the goal is to
ensure a given alternative's victory (and we focus on this variant of
the problem as well). The destructive variant, where the goal is to
prevent someone from winning, was introduced by \citet{HHR07}. These
papers focused on the Plurality rule and the Condorcet rule (and the Approval
rule, for the destructive case of \citet{HHR07}). Since then, many other
researchers extended this study to a number of other rules and
models~\cite{BU09,FHHR09,FHH11,FHH13,LFZL09,LZ10,MPRZ08,PX12}.

In all previous work on election control, the authors always assumed
that one could affect each entity of the election at unit cost only.
For example, 
one could add a voter at a unit cost and adding two voters always
was twice as expensive as adding a single voter. Only the paper of
\citet{FHH13}, where the authors study control in weighted elections,
could be seen as an exception: One could think of adding a voter of
weight $w$ as adding a group of $w$ voters of unit weight. On the one
hand, the weighted election model does not allow one to express rich
combinatorial structures as those that we study here, and on the other
hand, in our study we consider unweighted elections only (though
adding weights to our model would be seamless).

The specific combinatorial flavor of our model has been inspired by
the seminal work of \citet{RMPAHR98Z}\footnote{According to google scholar, accessed April~2014, cited more than 1000~times.}
on \emph{combinatorial auctions} (see, e.g., \citet{San06} for additional
information).  
There, bidders can place bids on combinations of items such that the bid on the combination of a set of items might be less than, equal to, or greater than the sum of the individual bids on each element from the same set of items.  
While in combinatorial auctions one ``bundles'' items to bid on, in our scenario one bundles voters.

In the computational social choice literature, combinatorial voting is
typically associated with scenarios where voters express opinions over
a set of items that themselves have a specific combinatorial structure
(typically, one uses CP-nets to model preferences over such
alternative sets~\cite{BBDHP04}). For example, \citet{CLX09} studied a
form of control in this setting and \citet{NMFK12Z} studied bribery problems.  
In contrast, we use the standard model of elections where
all alternatives and preference orders are given explicitly, but we
have a combinatorial structure of the sets of voters that can be
added.

\section{Preliminaries}\label{sec:pre}
We assume familiarity with standard notions regarding algorithms and
complexity theory. For each nonnegative integer $z$, we write $[z]$ to
mean $\{1, \ldots, z\}$.

\paragraph{Elections.}
An election~$E:=(\electionC,\electionV)$ consists of 
a set~$C$ of $m$~alternatives and a set~$V$ of $|V|$~voters~$v_1,v_2,\ldots, v_{|V|}$. 
Each voter~$v$ has a linear order~$\pref_{v}$ over the set~$\electionC$,
which we call a \emph{preference order}. 
For example, let~$\electionC=\{c_1,c_2,c_3\}$ be a set of alternatives.
The preference order~$c_1\pref_{v} c_2 \pref_{v} c_3$ of voter~$v$ indicates that
$v$ likes $c_1$ the best~($1$\textsuperscript{st} position), then~$c_2$, 
and $c_3$ the least~($3$\textsuperscript{rd} position).
We call a voter~$v\in \electionV$ a \emph{\votertype{c}}
if $c$ is at the first position of her preference order.
Given a subset~$C'\subseteq C$ of alternatives, if not stated explicitly,
we write $\seq{C'}$ to denote an arbitrary but fixed preference order over~$C'$. 

\paragraph{Voting Rules.}
A voting rule $\calR$ is a function that given an election~$E$ outputs
a (possibly empty) set $\calR(E) \subseteq C$ of the (tied) election
winners.  
We study the Plurality rule and the Condorcet rule.  Given an election, the
\emph{Plurality score} of an alternative~$c$ is the number of voters
that have $c$ at the first position in their preference orders; an
alternative is a Plurality winner if it has the maximum Plurality
score. An alternative $c$ is a \emph{Condorcet
  winner}~\cite{dC85} 
if it beats all other alternatives in
head-to-head contests. That is, $c$ is a \emph{Condorcet winner} in
election~$E=(\electionC, \electionV)$ if for each alternative $c' \in
\electionC \setminus \{c\}$ it holds that
$
|\{v \in \electionV \mid c \pref_{v} c'\}| > |\{v\in \electionV \mid c' \pref_v c\}|.
$
Condorcet's rule elects the (unique) Condorcet winner if it exists,
and returns an empty set otherwise. A voting rule is
\emph{Condorcet-consistent} if it elects a Condorcet winner when there is one
(however, if there is no Condorcet winner, then a Condorcet-consistent
rule is free to provide any set of winners).

\paragraph{Domain Restrictions.}
Intuitively, an election is \emph{single-peaked}~\cite{Black1948} if it is
possible to order the alternatives on a line in such a way that for
each voter~$v$ the following holds: If $c$ is $v$'s most preferred
alternative, then for each two alternatives $c_i$ and $c_j$ that both are
on the same side of $c$ (with respect to the ordering of the
alternatives on the line), among $c_i$ and $c_j$, $v$ prefers the one
closer to $c$. For example, single-peaked elections arise when we view
the alternatives on the standard political left-right spectrum and
voters form their preferences based solely on alternatives' positions
on this spectrum. 
Formally, we have the following definition.
\begin{definition}[Single-peaked elections]
  Let $C$ be a set of alternative and let $L$ be a linear order over
  $C$ (referred to as the societal axis). We say that a preference
  order $\pref$ (over $C$) is \emph{single-peaked} with respect to $L$ if for
  each three alternative $x,y,z \in C$ it holds that:
  \[
  \left( (x \mathrel{L} y \mathrel{L} z) \lor (z \mathrel{L} y \mathrel{L} x) \right)
  \implies \left( (x \pref y) \implies (y \pref z) \right).
  \]
  An election $(C,V)$ is single-peaked with respect to $L$ if the preference order of each
  voter in $V$ is single-peaked with respect to $L$. An election is
  single-peaked if there is a societal axis with respect to which it
  is single-peaked.
\end{definition}
There are polynomial-time algorithms that given an election decide if
it is single-peaked and, if so, provide a societal axis for
it~\cite{BT86,ELO08}.
\emph{Single-crossing} elections, introduced by \citet{Roberts1977}, capture
a similar idea as single-peaked ones, but from a different
perspective. This time we assume that it is possible to order the
voters so that for each two alternatives $a$ and~$b$ either all voters
rank $a$ and $b$ identically, or there is a single point along this
order where voters switch from preferring one of the alternatives to
preferring the other one. Formally, we have the following definition.

\begin{definition}[Single-crossing elections]
  An election $E = (C,V)$ is \emph{single-crossing} if there is
  an order $L$ over $V$ such that for each two alternatives $x$ and
  $y$ and each three voters $v_1, v_2, v_3$ such that $v_1 \mathrel{L}
  v_2 \mathrel{L} v_3$ it holds that:
  \[ (x \pref_{v_1} y \land x \pref_{v_3} y) \implies x \pref_{v_2}
  y.\]
\end{definition}
As for the case of single-peakedness, 
there are polynomial-time
algorithms that decide if an election is single-crossing and,
if so,
produce the voter order witnessing this fact~\cite{EFS12,BCW13}.

\paragraph{Combinatorial Bundling Functions.}
Given a voter set~$X$, a combinatorial \assignment
$\combRule: X\to 2^{X}$ (abbreviated as \emph{\assignmentnoxspace}) is a
function assigning to each voter a subset of voters. For convenience,
for each subset $X' \subseteq X$, we let $\combRule(X') = \bigcup_{x
  \in X'}\combRule(x)$.  For $x \in X$, $\combRule(x)$ is called $x$'s
\emph{bundle} (and for this bundle, $x$ is called its \emph{leader}).
We assume that $x \in \combRule(x)$ and so $\combRule(x)$ is never
empty. We typically write $b$ to denote the maximum bundle size under
a given~$\combRule$ (which will always be clear from context).
Intuitively, we use combinatorial \assignments to describe the sets of
voters that we can add to an election at a unit cost. For example, one
can think of $\combRule(x)$ as the group of voters that join the election under $x$'s influence.
We represent \assignments explicitly: For each voter $x$ we list
the voters in~$\combRule(x)$.

We are interested in various special cases of \assignments.
We say that $\combRule$ is \emph{\leaderanonymous} if for
each two voters $x$ and $y$ with the same preference order $\combRule(x) = \combRule(y)$ holds.  
Furthermore, $\combRule$ is \emph{\followeranonymous} if for each two voters $x$ and $y$
with the same preference orders, and each voter~$z$,
it holds that $x \in \combRule(z)$ if and only if $y \in \combRule(z)$.
We call $\combRule$ \emph{anonymous} if it is both \leaderanonymous
and \followeranonymous. 
One possible way of thinking about an \anonassignment is that 
it is a function assigning to each preference order appearing in the input 
a subset of the preference orders appearing in the input. 
For example, \anonassignments naturally model scenarios such as airing TV advertisements that appeal to particular groups of voters.

The swap distance between two voters $v_i$ and $v_j$ is the minimum number
of swaps of consecutive alternatives that transform $v_i$'s preference
order into that of~$v_j$.
Given a number~$d\in \mathbb{N}$, we
call $\combRule$ a \emph{\fullassignment} if for each $x\in X$,
$\combRule(x)$ is exactly the set of all $y\in X$ such that the swap
distance between the preference orders of $x$ and $y$ is at most~$d$.

  We introduce the concept of a \emph{\bundlinggraph} of an election,
  which, roughly speaking, models how the bundles of two voters interact with each other.

  \begin{definition}[\Bundlinggraph{s}]\label{def:bundlinggraph}
    Given an input instance to \probCCCAV,
    the \emph{\bundlinggraph} is a simple and directed graph $G = (V(G), E(G))$.
    For each voter $x$ there is a vertex $u_x \in V(G)$,
    and for each two distinct voters $y$ and $z$ such that $y \in \combRule(z)$ there is an arc $(u_z \to u_y) \in E(G)$.
  \end{definition}

  For arbitrary \assignments, the \bundlinggraph is a directed graph.
  However, if \combRule is a \fullassignment, that is, for each voter~$v$,
  $\combRule(v)$ contains all the voters at swap distance~$d$,
  then the \bundlinggraph can be thought of as being undirected,
  due to the following.
    
  \begin{lemma}\label{lem:x_in_y_y_in_x}
    If \combRule is a \fullassignment,
    then for any unregistered voter $x$ and any $y \in \combRule(x)$,
    it holds that $x \in \combRule(y)$.
  \end{lemma}

  \begin{proof}
    To see why the statement holds, notice that for any two voters $x$ and $y$,
    if $y \in \combRule(x)$, then the \swapdist between $x$ and $y$ is at most $d$,  therefore, because $\combRule$ is a \fullassignment, $x$ must be in $\combRule(y)$.
    This implies that for any arc $(u_x \to u_y)$ in the \bundlinggraph,
    the corresponding arc $(u_y \to u_x)$ is also present in the \bundlinggraph,
    therefore,
    we can treat the \bundlinggraph as an undirected graph.
  \end{proof}

Notice that this is not always the case for an arbitrary \assignment.
For instance, $\combRule(x) = \{x, y\}$, $\combRule(y) = \{y\}$ is a valid possibility for a bundling function.

\paragraph{Central Problem.}
We consider the following problem for a given voting rule~$\calR$:
\begin{quote}
  $\calR$ \probCCCAVLong \\ ($\calR$-\probCCCAV) \\
  \textbf{Input:} An election~$E=(\electionC, \electionV)$,
  a set~$\electionW$ of (unregistered) voters with $\electionV\cap \electionW = \emptyset$,
  a \assignment $\combRule: \electionW \to 2^{\electionW}$, a
  preferred alternative~$p~\in~\electionC$,
  and a bound~$k \in \mathbb{N}$.\\
  \textbf{Question:} Is there a subset of voters~$\electionW'\subseteq
  \electionW$ of size at most $k$ such that $p \in \calR(\electionC, \electionV \cup \combRule(\electionW'))$,
  where $\calR(C,X)$ is the set of winners of the election $(C, X)$ under the rule $\calR$ ?
\end{quote}
We note that we use here a so-called nonunique-winner model. For a
control action to be successful, it suffices for $p$ to be one of the
tied winners. 
Throughout this work, we refer to the set $\electionW'$ of voters such that $p$ wins election $(\electionC, \electionV \cup \combRule(\electionW'))$ as the solution
and denote $k$ as the solution size. 


$\calR$-\probCCCAV is a generalization of the well-studied
problem~$\calR$ \probCCAVLong~($\calR$-\probCCAV) (in which
$\combRule$ is fixed so that for each $w \in W$ we have $\combRule(w)
= \{w\}$).  The non-combinatorial problem \probCCAV is polynomial-time
solvable for the Plurality rule~\cite{BTT92}, but is $\np$-complete for the
Condorcet rule~\cite{LFZL09}, therefore:

\begin{observation}\label{obs:condorcetIsHard}
  \textsc{Condorcet}-\probCCCAV is $\np$-hard even 
  if the maximum bundle size~$b$ is one.
\end{observation}


\paragraph{Parameterized Complexity.}
An instance~$(I,k)$ of a parameterized problem consists of the actual
instance~$I$ and an integer~$k$ being the \emph{parameter}~\cite{DF13,FG06,Nie06}.
A parameterized problem is called \emph{fixed-parameter tractable} (is in $\fpt$)
if there is an algorithm solving it in~$f(k)\cdot|I|^{O(1)}$ time,
for an arbitrary computable function $f$ only depending on parameter~$k$,
whereas an algorithm with running-time~$|I|^{f(k)}$ only shows membership in the class~$\xp$ (clearly, $\fpt \subseteq \xp$).
If a parameterized problem is fixed-parameter tractable due to a formulation as integer linear program (ILP), then we say that this problem is in $\ilpfpt$.
One can show that a parameterized problem~$L$ is (presumably) not
fixed-parameter tractable by devising a \emph{parameterized reduction} from a $\wone$-hard or a $\wtwo$-hard problem (such as \probClique or \probSetCover parameterized by the ``solution size'') 
to~$L$.
A parameterized reduction from a parameterized problem~$L$ to another parameterized problem~$L'$
is a function that,
given an instance~$(I,k)$,
computes in~$f(k)\cdot |I|^{O(1)}$ time an instance~$(I',k')$,
such that~$k' \le g (k)$
and~$(I,k)\in L \Leftrightarrow (I',k')\in L'$.
\citet{BBCN12} survey 
parameterized complexity investigations in voting.
\begin{table}[t!]

  \centering
   \begin{tabular}[t!]{|l|c|c|c|c|c|}
    \hline 
    & $m$ & $n$ & $k$ & $b$ & $d$ \\
    \hline \hline
    & \multicolumn{5}{c|}{} \\[-1em]
    \multirow{2}{*}{\# alternatives ($m$)} & 
    \multicolumn{5}{c|}{Non-anonymous: $\wtwo$-h wrt.~$k$ even if $m=2$~\mytabref{thm:CCCAVw2-k-m=2}}\\
    &       \multicolumn{5}{c|}{Anonymous: $\ilpfpt$ wrt.~$m$~\mytabref{thm:CCCAVisFPTm}\phantom{Non---}} \\
    \hline
    &\cellcolor{lightgray}& \multicolumn{4}{c|}{} \\[-1em]
    \# unreg.\ voters ($n$) & \cellcolor{lightgray}& \multicolumn{4}{c|}{$\fpt$ wrt.~$n$} \\
    \cline{1-1}\cline{2-6}
    & \multicolumn{2}{c|}{\cellcolor{lightgray}} & \multicolumn{2}{c|}{} & \\[-1em]
    \multirow{3}{*}{solution size ($k$)} & \multicolumn{2}{c|}{\cellcolor{lightgray}} & \multicolumn{2}{c|}{$\xp$ \mytabobs{obs:CCCAVinXPk}} &  \multirow{3}{4cm}{Single-peaked \& full-$1$ $\combRule$: $\wone$-h wrt.~$k$~\mytabref{thm:SinglepeakedWhardness}}\\
                                         & \multicolumn{2}{c|}{\cellcolor{lightgray}} & \multicolumn{2}{c|}{Anonymous \& $b=3$:} & \\
                                         & \multicolumn{2}{c|}{\cellcolor{lightgray}} & \multicolumn{2}{c|}{$\wone$-h wrt.~$k$~\mytabref{thm:CCCAVw1k-anonymous-b=3}} &   \\
    \cline{1-1}\cline{4-6}
     & \multicolumn{3}{c|}{\cellcolor{lightgray}} & \multicolumn{2}{c|}{} \\[-1em]
     max. bundle size ($b$) & \multicolumn{3}{c|}{\cellcolor{lightgray}} & 
     \multicolumn{2}{c|}{\parbox[l]{8cm}{
                                         $b=2:$  $\np$-h~\mytabref{thm:np-h-b=2} and $\p$ for full-$d$ $\combRule$~\mytabref{obs:b2fullp}\\
                                         $b=3:$ $\np$-h even for full-$d$ $\combRule$ \mytabref{thm:fullB3HARD}\\
                                         $b\ge 4:$ $\np$-h even for full-$1$ $\combRule$~\mytabref{prop:np-hard:d=1:b=4}}}\\
    \cline{1-1}\cline{5-6}
    \multirow{2}{*}{max.\ swap dist.\ ($d$)} & \multicolumn{4}{c|}{\cellcolor{lightgray}} & $d=1$: $\wone$-h wrt.~$k$~\mytabref{thm:SinglepeakedWhardness}\\
     & \multicolumn{4}{c|}{\cellcolor{lightgray}} & Single-crossing \& full-$d$ $\combRule$: $\p$ \mytabref{thm:cccav-P-sc}\\
    \hline
    \multicolumn{6}{c}{}\\[-1em]
  \end{tabular}

  \caption{
    \textbf{Computational complexity classification of \textsc{Plurality}-\probCCCAV}
    (since the non-combinatorial problem \probCCAV is already $\np$-hard for Condorcet's rule, we concentrate here on the Plurality rule).
    Each row and column in the table corresponds to a parameter 
    such that each cell contains results for the two corresponding parameters combined.
    Due to symmetry, 
    there is no need to consider the cells under the main diagonal,
    therefore they are painted in gray. 
    $\ilpfpt$ means $\fpt$ based on a formulation as an integer linear program.
}
  \label{tab:result}
\end{table}

\paragraph{Our Contributions.}
We introduce a new model for combinatorial control in voting.  
As $\calR$-\probCCCAV is generally $\np$-hard even for $\calR$ being the Plurality rule, we show several fixed-parameter
tractability results for some of the natural parameterizations of
$\calR$-\probCCCAV; we almost completely resolve the complexity of \probCCCAV,
for the Plurality rule and the Condorcet rule,
as a function of the maximum bundle size $b$ and the maximum distance
$d$ from a voter~$v$ to the farthest element of her bundle.
Further, we show that the problem remains hard even when restricting
the elections to be single-peaked, but that it is polynomial-time
solvable when we focus on single-crossing elections. Our results for
Plurality elections are summarized in \cref{tab:result}.

\section{Complexity for Unrestricted Elections}
\label{sec:unrestricted elections}
In this section we provide our results for the case of unrestricted
elections, where voters may have arbitrary preference orders. In the
next section we will consider single-peaked and single-crossing
elections that only allow ``reasonable'' preference orders.

\subsection{Number of Voters, Number of Alternatives, and Solution Size}
We start our discussion by considering parameters
``the number $m$ of alternatives'',
``the number $n$ of unregistered voters'',
and ``the solution size $k$''. 
A simple brute-force algorithm, checking all possible combinations of $k$ bundles,
proves that both \textsc{Plurality}-\probCCCAV and
\textsc{Condorcet}-\probCCCAV are in $\xp$ for parameter~$k$, and
in $\fpt$ for parameter $n$ (the latter holds because $k \leq n$). Indeed, the same result holds for all voting rules that
are $\xp/\fpt$-time computable for the respective parameters.

\begin{observation}\label{obs:CCCAVinXPk}
  \obsCCCAVinXPk
\end{observation}

The $\xp$ result for \textsc{Plurality}-\probCCCAV
with respect to the parameter $k$ probably cannot be improved to
fixed-parameter tractability. 
Indeed, for parameter $k$ we show that
the problem is $\wone$-hard, even for \anonassignments and for maximum bundle size three.

\begin{theorem}\label{thm:CCCAVw1k-anonymous-b=3}
  \thmwonekanonbthree
\end{theorem}



\begin{proof}
  We provide a parameterized reduction from the $\wone$-hard problem \probClique parameterized by the parameter~$h$~\cite{DF13},
  which asks for the existence of a size-$h$ clique in an input graph~$G$.

  \probDef {\probClique} {An undirected graph $G=(V(G), E(G))$ and $h
    \in \mathbb{N}$.}  {Does $G$ admit a size-$h$ \emph{clique}, that
    is, a size-$h$ vertex subset~$U\subseteq V(G)$ such that $G[U]$ is
    complete?}
  
  Let $(G, h)$ be a \probClique instance. Without loss of generality,
  we assume that $G$ is connected, that $h \geq 3$, and that each vertex in $G$ has degree at
  least $h-1$.  We construct an election~$E=(\electionC, \electionV)$
  with $\electionC:=\{p, w, g\} \cup \{c_e \mid e\in E(G)\}$.  The
  registered voter set~$\electionV$ consists of $\binom{h}{2} + h$ voters each with preference order~$w \pref \seq{C\setminus \{w\}}$,
  another $\binom{h}{2}$ voters each with preference order~$g \pref \seq{C \setminus \{g\}}$, and
  another $h$ voters each with preference order~$p \pref \seq{C\setminus \{p\}}$.
  For each vertex $u \in V(G)$, we define $C(u):= \{c_e \mid e \in E(G) \wedge u \in e\}$,
  and construct the set~$\electionW$ of unregistered voters as follows:
  \begin{enumerate}[(i)]
  \item For each vertex $u \in V(G)$, we add an unregistered \votertype{g}~$w_{u}$ with preference
    order~$g \pref \seq{C(u)} \pref \seq{C \setminus
      (\{g\} \cup C(u))}$, and we set $\combRule(w_u) = \{w_u\}$.
    We call these unregistered voters \emph{vertex voters}.    
  \item For each edge~$e = \{u,u'\} \in E(G)$, 
    we add an unregistered \votertype{p}~$w_e$ with preference order~$p \pref c_e \pref \seq{C \setminus \{p, c_e\}}$,
    and we set $\combRule(w_e) = \{w_u, w_{u'}, w_e \}$.
    We call these unregistered voters \emph{edge voters}.
  \end{enumerate}
  Since all the unregistered voters have different preference orders (this
  is so because $G$ is connected, $h \geq 3$, and each vertex has degree at least $h-1$),
  every \assignment for our instance, $\combRule$ included, is anonymous.
  To finalize our construction, we set $k := \binom{h}{2}$. 

  We show that $G$ has a size-$h$ clique if and only if $(E =
  (\electionC, \electionV), \electionW, \combRule, p, k)$ is a yes instance
  for \textsc{Plurality}-\probCCCAV.
  For the ``if'' part, suppose that there is a subset~$W'$ of at most
  $k$ voters such that $p$ wins the election~$(\electionC, \electionV \cup \combRule(\electionW'))$.  
  We show that the vertex set~$U':=\{u\in V(G)
  \mid w_e \in W' \land u \in e\}$ is a size-$h$ clique for $G$.
  First, we observe that $p$ needs at least $\binom{h}{2}$~points to
  become a winner because of the difference in scores between the
  initial winner $w$ and $p$.  By our construction, only bundles that
  include the edge voters give points to $p$ and each of such bundles gives $p$ exactly one point. 
  Since we can add at most
  $k = \binom{h}{2}$ bundles, we must add exactly $k$ bundles of the
  edge voters.  This means that $E(G[U'])$ contains at least $k$
  edges.
  However, in order to ensure $p$'s victory, $\combRule(W')$ may only
  give at most $h$ additional points to $g$.  This means that $U'$
  contains at most $h$ vertices. 
  With $|E(G[U'])|\ge k$, we conclude that
  $U'$ is of size $h$ and, hence, is a size-$h$ clique for $G$.

  For the ``only if'' part, suppose that $U'\subseteq V(G)$ is a size-$h$ clique for $G$. 
  We construct the subset~$W'$ by adding to it any edge voter~$w_e$ with $e \in E(G[U'])$. 
  Obviously, $|W'| = k$.
  Now it easy to check that $p$ co-wins with both $w$ and $g$ the election $(\electionC, \electionV \cup \combRule(W'))$ with score~$\binom{h}{2}+h+1$.
\end{proof}

If we drop the anonymity requirement for the \assignment,
then we obtain a stronger intractability result. For parameter $k$,
the problem becomes $\wtwo$-hard, even for two alternatives.
This is quite remarkable because typically election problems with a
small number of alternatives are easy (they can be solved either through
brute-force attacks or through integer linear programming attacks
employing the famous $\fpt$ algorithm of Lenstra~\cite{Len83}; 
see 
the survey of \citet{BBCN12} for examples, but
note 
that there are also known examples of problems where a small
number of alternatives does not seem to help~\cite{BCFNN14}). Further,
since our proof uses only two alternatives, it applies to almost all
natural voting rules: For two alternatives almost all of them (including
the Condorcet rule) are equivalent to the Plurality rule.
The reduction is from the $\wtwo$-complete problem~\probSetCover parameterized by the solution size~\cite{DF13}.

\begin{theorem}\label{thm:CCCAVw2-k-m=2}
  \thmwtwokmtwo
\end{theorem}

\begin{proof}
    We provide a parameterized reduction from the $\wtwo$-complete problem~\probSetCover parameterized by the parameter~$h$~\cite{DF13}.
  \probDef
    {\probSetCover}
    {A collection~$\mathcal{S} = \{S_1,\ldots,S_m\}$ of subsets of the universe~$X = \{x_1,\ldots,x_n\}$ and $h \in \mathbb{N}$.}
    {Is there a size-$h$ subset $\mathcal{S'} \subseteq \mathcal{S}$ that covers the universe, that is, $\bigcup \mathcal{S'} = X$?}

  Let $(S,X,h)$ be a \probSetCover instance.
  We construct an election $E = (C, V)$ with $C = \{p, g\}$.
  The registered voter set $V$ consists of only $(n - k)$ \voterstype{g}.
  We construct the unregistered voter set $W$ as follows:
    
  \newcommand{\setcoverelementvoter}{\emph{element-voter}\xspace}
  \newcommand{\setcoverelementvoters}{\emph{element-voters}\xspace}
  \newcommand{\setcoversetvoter}{\emph{set-voter}\xspace}
  \newcommand{\setcoversetvoters}{\emph{set-voters}\xspace}
  \newcommand{\setcoverdummyvoter}{\emph{dummy-voter}\xspace}
  \newcommand{\setcoverdummyvoters}{\emph{dummy-voters}\xspace}
    
  \begin{enumerate}[(i)]
    \item For each element $x_i\in X$, we construct one \votertype{p}, denoted by $w^{x}_{i}$
    (called \setcoverelementvoter),
    and two \voterstype{g}, denoted by $w^{d}_{i^1}$ and $w^{d}_{i^2}$
    (called \setcoverdummyvoters),
    and we set $\combRule(w^{x}_{i}) = \{w^{x}_{i}, w^{d}_{i^1}, w^{d}_{i^2}\}$ and 
    set $\combRule(w^{d}_{i^1}) = \combRule(w^{d}_{i^2}) = \{w^{d}_{i^1}, w^{d}_{i^2}\}$.
    \item For each set~$S_j$, we construct one \votertype{g}, denoted by $w^{S}_{_j}$
    (called \setcoversetvoter),
    and we set $\combRule(w^{S}_{j}) = \{w^{S}_{j}\} \cup \{w^{x}_{i} \mid x_i \in S_j\}$.
    That is, the bundle for the voter corresponding to a set contains all of the voters corresponding to the elements of the set.

  \end{enumerate}
  Finally, we set $k := h$ and let $d$ be arbitrary.
  
  The construction is obviously a parameterized reduction,
  and we show now that there is a size-$h$ subset $\mathcal{S}' \subseteq \mathcal{S}$ that covers the universe
  if and only if there is a size-$k$ subset $\electionW'$ of unregistered voters,
  such that if added (with their respective bundles) to the election,
  $p$ becomes a Plurality winner of the election.
  
  For the ``if'' part, suppose that there is such a size-$k$ subset $W'$.
  Also,
  if there are some \setcoverelementvoters in the solution,
  then we can simply remove them,
  as they do not help $p$ win,
  due to the \setcoverdummyvoters.
  The only way to achieve the score increase of $n - k$ for $p$
  is to have all of the \setcoverelementvoters added to the election,
  and this can be done only by covering all of the universe,
  with at most $k$ \setcoversetvoters;
  therefore,
  the solution corresponds to a set covering of the universe.
  
  For the ``only if'' part,
  given a size-$h$ subset $\mathcal{S}' \subseteq \mathcal{S}$ that covers the universe,
  we choose,
  for every $S_j \in S'$,
  its respective voter $w^S_{j}$,
  and add it to the election.
  This gives a size-$k$ subset $\electionW'$ of unregistered voters,
  which easily can be verified to result in $p$ winning the election.
  
  As for the Condorcet rule, 
  we use the same unregistered voters as defined above 
  and we construct the original election with
  ($2n - k - 1$) \voterstype{g} and ($n - k$)~\voterstype{p}.
\end{proof}

The above proof uses the non-anonymity of the \assignment in a crucial way.
If we require the \assignment to be anonymous,
then \probCCCAV can be formulated as an integer linear program
where the number of variables and the number of constraints are bounded by some function in
the number $m$ of alternatives.
The idea behind this is that with anonymity we can formulate our problem through an integer linear program 
where the number of variables and the number of constraints are bounded by some function in $m$.
Such integer linear programs are in $\fpt$ with respect to the number of variables~\cite{Len83}.

\begin{theorem}\label{thm:CCCAVisFPTm}
  \thmfptm
\end{theorem}

\begin{proof}
  \newcommand{\numvoters}[1]{
    \ensuremath{N_{#1}}
  }
  \newcommand{\bundle}[2]{
  \ensuremath{M_{#1}^{#2}}}

  \newcommand{\varvoters}[1]{
    \ensuremath{x_{#1}}
  }

  \newcommand{\firstpos}[2]{
    \ensuremath{B_{#1}^{#2}}
  }

  \newcommand{\s}[1]{
    \ensuremath{s(#1)}
  }

  We describe an integer linear program~(ILP) with at most $m!$ variables and at most $m! + m$ constraints that solves both \textsc{Plurality}-\probCCCAV and \textsc{Condorcet}-\probCCCAV. 
  Fixed-parameter tractability then follows,
  since any ILP with $\rho$ variables and $L$ input bits is solvable in $O(\rho^{2.5\rho+o(\rho)}L)$ time~(\cite{Len83} and~\cite{Kan87}).

  With $m$~alternatives, there are at most $m!$ voters with pairwise different preference orders in a given election.
  For each alternative $a\in \electionC$, let $\s{a}$ be its initial score.
  Since the voters are anonymous, there are at most $m!$ different bundles.
  Furthermore, we can assume that all voters in $W'$ have pairwise different preference orders
  (this is because, due to anonymity, there is no additional gain of adding two voters with the same preference order).
  
  Let $\pref_1, \pref_2, \ldots, \pref_{m!}$ be an ordering of all of the possible preference orders over $m$ alternatives.
  For $i \in [m!]$, let $\numvoters{i}$ be the number of voters with preference order~$\pref_i$ in $W$.
  For $i \in [m!]$ and $j \in [m!]$, let $\bundle{i}{j}$ have value $1$ 
  if there is a voter with preference order~$\pref_j$ 
  that is in the bundle of a voter whose preference order is $\pref_i$, and otherwise $0$.
  For each alternative~$a$ and each $i \in [m!]$,
  let $\firstpos{i}{a}=1$ if alternative~$a$ is at the first position in the preference order~$\pref_i$
  (that is, $i$ is a \votertype{a}), and otherwise $\firstpos{i}{a}=0$.
  
  For each preference order~$\pref_i$, $i \in [m!]$, 
  we introduce one boolean variable~$\varvoters{i}$,
  with the intent that the value of $\varvoters{i}$ will be $1$ if and only if $W'$ contains a voter with preference order $i$.
  Indeed,
  an integer linear program usually tries to minimize or maximize a certain function,
  while here,
  we write the integer linear program as simply a feasibility problem.
  It can be easily rewritten with a minimization function instead.
  Now we are ready to state the integer linear program. 


  \begin{flalign}
    \label{ilp:solutionbound}
    \sum_{i\in [m!]}{\varvoters{i}} &\le k, & \\
    \label{ilp:voterbound}
    \varvoters{i} &\le\numvoters{i} & \forall i \in [m!],\quad\quad \\
    \label{ilp:winning}
    \sum_{i \in [m!]}\sum_{j \in [m!]}(\firstpos{j}{a}-\firstpos{j}{p})\cdot 
    \numvoters{j} \cdot \bundle{i}{j} \cdot \varvoters{i} &< s(p) - s(a) & \forall a \in C\setminus \{p\}.\quad\quad
  \end{flalign}
  
  Constraint~(\ref{ilp:solutionbound}) ensures that at most $k$ voters are added to $W'$.
  Constraint~(\ref{ilp:voterbound}) ensures that the voters added to $W'$ are available in $W$.
  Constraint~(\ref{ilp:winning}) ensures that no other alternative has a higher Plurality score than alternative~$p$.
  It can be easily verified that there is a solution for this integer linear program
  if and only if there is a solution to the input instance.
\end{proof}

\subsection{Combinatorial Parameters}



We focus now on the complexity of \textsc{Plurality}-\probCCCAV as a
function of two combinatorial parameters: \begin{inparaenum}[(a)]
  \item the maximum swap distance $d$ between the leader and his followers in one bundle, and
  \item the maximum size $b$ of each voter's bundle.
\end{inparaenum}

Specifically, we show that if \combRule is a \nonfullassignment (that is, it is not required to contain all voters at a given distance),
then Plurality-\probCCCAV is polynomial-time solvable if the maximum bundle size is one,
but if the maximum bundle size is two, then Plurality-\probCCCAV is $\np$-hard.
However, if \combRule is a \fullassignment (that is, if is required to contain all voters at a given distance),
then Plurality-\probCCCAV is polynomial-time solvable if the maximum bundle size is two,
but if the maximum bundle size is three, then Plurality-\probCCCAV is $\np$-hard.

First, if $b = 1$, then \probCCCAV reduces to \probCCAV
and, thus, can be solved by a greedy algorithm in polynomial time~\cite{BTT92}.
\begin{observation}\label{obs:b1p}
 If the maximum bundle size~$b$ is one,
 then \textsc{Plurality}-\probCCCAV is polynomial-time solvable.
\end{observation}

However, for arbitrary \assignments, \textsc{Plu\-rality}-\probCCCAV becomes intractable as soon as $b = 2$.

\begin{theorem}\label{thm:np-h-b=2}
  \thmnphbtwo
\end{theorem}

\begin{proof}
  \begin{figure}[t!]
    \centering
    
    \begin{subfigure}[a]{.48\textwidth}
      \begin{tikzpicture}[every node/.style={draw=black,thick,circle,inner sep=0pt}]
        \def \n {8}
        \def \radius {2.5cm}
        \def \margin {12}
        \def \nodesize {0.85cm}
        
        \def \s {1}
        \node[draw, circle,fill=black!30][minimum size=\nodesize] at ({360/\n * (\s - 1)}:\radius) {$c_r^{x_j}$};
        \draw[<-, >=latex] ({360/\n * (\s - 1)+\margin}:\radius) 
        arc ({360/\n * (\s - 1)+\margin}:{360/\n * (\s)-\margin}:\radius);

        \def \s {2}
        \node[draw, circle][minimum size=\nodesize] at ({360/\n * (\s - 1)}:\radius) {$p^j_2$};
        \draw[<-, >=latex] ({360/\n * (\s - 1)+\margin}:\radius) 
        arc ({360/\n * (\s - 1)+\margin}:{360/\n * (\s)-\margin}:\radius);

        \def \s {3}
        \node[draw, circle,fill=black!30][minimum size=\nodesize] at ({360/\n * (\s - 1)}:\radius) {$c_i^{\lnot x_j}$};
        \draw[<-, >=latex] ({360/\n * (\s - 1)+\margin}:\radius) 
        arc ({360/\n * (\s - 1)+\margin}:{360/\n * (\s)-\margin}:\radius);

        \def \s {4}
        \node[draw, circle][minimum size=\nodesize] at ({360/\n * (\s - 1)}:\radius) {$p^j_1$};
        \draw[<-, >=latex] ({360/\n * (\s - 1)+\margin}:\radius) 
        arc ({360/\n * (\s - 1)+\margin}:{360/\n * (\s)-\margin}:\radius);
        
        \def \s {5}
        \node[draw, circle,fill=black!30][minimum size=\nodesize] at ({360/\n * (\s - 1)}:\radius) {$c_t^{x_j}$};
        \draw[<-, >=latex] ({360/\n * (\s - 1)+\margin}:\radius) 
        arc ({360/\n * (\s - 1)+\margin}:{360/\n * (\s)-\margin}:\radius);

        \def \s {6}
        \node[draw, circle][minimum size=\nodesize] at ({360/\n * (\s - 1)}:\radius) {$p^j_4$};
        \draw[<-, >=latex] ({360/\n * (\s - 1)+\margin}:\radius) 
        arc ({360/\n * (\s - 1)+\margin}:{360/\n * (\s)-\margin}:\radius);
        
        \def \s {7}
        \node[draw, circle,fill=black!30][minimum size=\nodesize] at ({360/\n * (\s - 1)}:\radius) {$c_s^{\lnot x_j}$};
        \draw[<-, >=latex] ({360/\n * (\s - 1)+\margin}:\radius) 
        arc ({360/\n * (\s - 1)+\margin}:{360/\n * (\s)-\margin}:\radius);

        \def \s {8}
        \node[draw, circle][minimum size=\nodesize] at ({360/\n * (\s - 1)}:\radius) {$p^j_3$};
        \draw[<-, >=latex] ({360/\n * (\s - 1)+\margin}:\radius) 
        arc ({360/\n * (\s - 1)+\margin}:{360/\n * (\s)-\margin}:\radius);
      \end{tikzpicture}
      \caption{Gadget used in \cref{thm:np-h-b=2}}\label{fig:gadget b=2}
    \end{subfigure}
    \begin{subfigure}{.48\textwidth}
      \begin{tikzpicture}[every node/.style={draw=black,thick,circle,inner sep=0pt}]
        \def \n {12}
        \def \radius {2.5cm}
        \def \margin {12}
        \def \nodesize {0.85cm}

        \def \s {1}
        \node[draw, circle,fill=black!30][minimum size=\nodesize] at ({360/\n * (\s - 1)}:\radius) {$c_r^{x_j}$};
        \draw[-, >=latex] ({360/\n * (\s - 1)+\margin}:\radius) 
        arc ({360/\n * (\s - 1)+\margin}:{360/\n * (\s)-\margin}:\radius);

        \def \s {2}
        \node[draw, circle][minimum size=\nodesize] at ({360/\n * (\s - 1)}:\radius) {$p^j_4$};
        \draw[-, >=latex] ({360/\n * (\s - 1)+\margin}:\radius) 
        arc ({360/\n * (\s - 1)+\margin}:{360/\n * (\s)-\margin}:\radius);

        \def \s {3}
        \node[draw, circle][minimum size=\nodesize] at ({360/\n * (\s - 1)}:\radius) {$p^j_3$};
        \draw[-, >=latex] ({360/\n * (\s - 1)+\margin}:\radius) 
        arc ({360/\n * (\s - 1)+\margin}:{360/\n * (\s)-\margin}:\radius);

        \def \s {4}
        \node[draw, circle,fill=black!30][minimum size=\nodesize] at ({360/\n * (\s - 1)}:\radius) {$c_i^{\lnot x_j}$};
        \draw[-, >=latex] ({360/\n * (\s - 1)+\margin}:\radius) 
        arc ({360/\n * (\s - 1)+\margin}:{360/\n * (\s)-\margin}:\radius);
        
        \def \s {5}
        \node[draw, circle][minimum size=\nodesize] at ({360/\n * (\s - 1)}:\radius) {$p^j_2$};
        \draw[-, >=latex] ({360/\n * (\s - 1)+\margin}:\radius) 
        arc ({360/\n * (\s - 1)+\margin}:{360/\n * (\s)-\margin}:\radius);

        \def \s {6}
        \node[draw, circle][minimum size=\nodesize] at ({360/\n * (\s - 1)}:\radius) {$p^j_1$};
        \draw[-, >=latex] ({360/\n * (\s - 1)+\margin}:\radius) 
        arc ({360/\n * (\s - 1)+\margin}:{360/\n * (\s)-\margin}:\radius);
        
        \def \s {7}
        \node[draw, circle,fill=black!30][minimum size=\nodesize] at ({360/\n * (\s - 1)}:\radius) {$c_t^{x_j}$};
        \draw[-, >=latex] ({360/\n * (\s - 1)+\margin}:\radius) 
        arc ({360/\n * (\s - 1)+\margin}:{360/\n * (\s)-\margin}:\radius);

        \def \s {8}
        \node[draw, circle][minimum size=\nodesize] at ({360/\n * (\s - 1)}:\radius) {$p^j_8$};
        \draw[-, >=latex] ({360/\n * (\s - 1)+\margin}:\radius) 
        arc ({360/\n * (\s - 1)+\margin}:{360/\n * (\s)-\margin}:\radius);

        \def \s {9}
        \node[draw, circle][minimum size=\nodesize] at ({360/\n * (\s - 1)}:\radius) {$p^j_7$};
        \draw[-, >=latex] ({360/\n * (\s - 1)+\margin}:\radius) 
        arc ({360/\n * (\s - 1)+\margin}:{360/\n * (\s)-\margin}:\radius);

        \def \s {10}
        \node[draw, circle,fill=black!30][minimum size=\nodesize] at ({360/\n * (\s - 1)}:\radius) {$c_s^{\lnot x_j}$};
        \draw[-, >=latex] ({360/\n * (\s - 1)+\margin}:\radius) 
        arc ({360/\n * (\s - 1)+\margin}:{360/\n * (\s)-\margin}:\radius);

        \def \s {11}
        \node[draw, circle][minimum size=\nodesize] at ({360/\n * (\s - 1)}:\radius) {$p^j_6$};
        \draw[-, >=latex] ({360/\n * (\s - 1)+\margin}:\radius) 
        arc ({360/\n * (\s - 1)+\margin}:{360/\n * (\s)-\margin}:\radius);

        \def \s {12}
        \node[draw, circle][minimum size=\nodesize] at ({360/\n * (\s - 1)}:\radius) {$p^j_5$};
        \draw[-, >=latex] ({360/\n * (\s - 1)+\margin}:\radius) 
        arc ({360/\n * (\s - 1)+\margin}:{360/\n * (\s)-\margin}:\radius);
      \end{tikzpicture}
      \caption{Gadget used in \cref{thm:fullB3HARD}}
      \label{fig:gadget b=3full}
    \end{subfigure}    
    \caption{Part of the construction used in \cref{thm:np-h-b=2} and \cref{thm:fullB3HARD}.
    Specifically, we show the cycle corresponding to variable~$x_j$ which occurs as a negative literal in clauses~$C_i$ and $C_s$
    and as a positive literal in clauses $C_{r}$ and $C_{t}$.}  
  \label{fig:ex-3SAT-VAR-gadget}
  \end{figure}
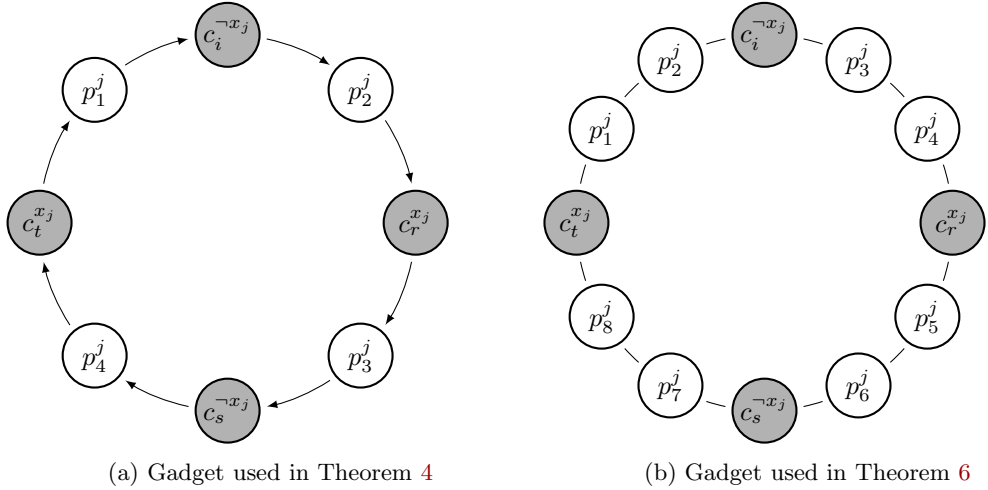

  We provide a reduction from a restricted variant of the $\np$-complete problem~\textsc{3SAT},
  where each clause has either two or three literals, 
  each variable occurs exactly four times,
  twice as a positive literal,
  and twice as a negative literal.

  \probDef
    {\probSAT}
    {A collection~$\mathcal{C}$ of clauses over the set~$\mathcal{X} = \{x_1, \ldots, x_n\}$ of variables such that each clause has either two or three literals,
    and each variable appears exactly four times, twice as a positive literal and twice as a negative literal.}
    {Is there a truth assignment that satisfies all the clauses in $\mathcal{C}$?}

  This variant is still $\np$-hard since from~\citet[Theorem 2.1]{Tov84}, 
  one obtained $\np$-hardness for \textsc{3SAT} where 
  each clause has either two or three literals, 
  each variable occurs either two or three times,
  and at most one time as a negative literal.

  We can reduce from this problem to \probSAT as follows.
  First, we assume that no variable appears only positively,
  because if this is the case,
  we can just set it to true and remove it.
  For each variable~$x_i$ that appears three times  
  (two times positively and one time negatively),
  we add one new variable~$y_i$,
  and two new clauses $\{\lnot x_i, \lnot y_i, \lnot y_i\}$ and $\{\lnot y_i, \lnot y_i\}$.
  For each variable~$x_i$ that appears two times
  (one time positively and one time negatively),
  we add one new clause~$\{\lnot x_i, x_i\}$.
  It can be verified that the original instance is a yes-instance if and only if the newly constructed instance is a yes-instance for \probSAT.

  Now, given a \probSAT instance~$(\mathcal{C}, \mathcal{X})$,
  where $\mathcal{C}$ is the set of clauses over the set of variables $\mathcal{X}$,
  we construct an election~$(\electionC, \electionV)$.
  We set $k:=4|\mathcal{X}|$,
  and construct the set~$\electionC$ of alternatives to be $\electionC:=\{p, w\} \cup \{c_i \mid C_i \in \mathcal{C}\}$,
  where the $c_i$ are called the \emph{clause} alternatives.
  We construct the set~$\electionV$ of registered voters such that the initial score of $w$ is $4|\mathcal{X}|$,
  the initial score of the clause alternative~$c_i$ is $4|\mathcal{X}| - |C_i| + 1$   (where $|C_i|$ is the number of literals that clause~$C_i$ contains),
  and the initial score of $p$ is zero. 
  We construct the set~$\electionW$ of unregistered voters as follows
  (throughout the rest of the proof, we will often write $\ell_j$ to refer to a literal that
  contains variable $x_j$; depending on the context, $\ell_j$ will mean either
  $x_j$ or $\lnot x_j$ and the exact meaning will always be clear):
  \begin{enumerate}
    \item for each variable~$x_j \in \mathcal{X}$,
    we construct four \voterstype{p}, denoted by $p^{j}_{1}, p^{j}_{2}, p^{j}_{3}, p^{j}_{4}$;
    we call such voters \emph{variable} voters.
    \item for each clause~$C_i \in \mathcal{C}$ and each literal~$\ell$ contained in $C_i$,
    we construct a \voterstype{c_i}, denoted by $c_i^{\ell}$;
    we call such voter a \emph{clause} voter. 
    Note that clause~$C_i$ has exactly $|C_i|$ corresponding clause voters.
  \end{enumerate}
  We define the assignment function $\combRule$ as follows: For each
  variable $x_j \in \mathcal{X}$ that occurs as a negative literal
  ($\lnot x_j$) in clauses $C_{i}$ and $C_{s}$, and as a positive
  literal ($x_j$) in clauses $C_{r}$ and $C_{t}$, we set
  \begin{alignat*}{2}
  \combRule(p^j_1) &:= \{p^j_1, c_i^{\lnot x_j}\}, &\quad \combRule(c_i^{\lnot x_j}) &:= \{c_i^{\lnot x_j}, p^j_2\},\\
  \combRule(p^j_2) &:= \{p^j_2, c_r^{x_j}\},&\quad \combRule(c_r^{x_j}) &:=
  \{c_r^{x_j}, p^j_3\},\\
  \combRule(p^j_3) &:= \{p^j_3, c_s^{\lnot
    x_j})\}, &\quad \combRule(c_s^{\lnot x_j}) &:= \{c_s^{\lnot x_j}, p^j_4\},\\
  \combRule(p^j_4) &:= \{p^j_4, c_t^{x_j}\}, &\quad \combRule(c_t^{x_j}) &:= \{c_t^{x_j}, p^j_1\}.
\end{alignat*}  

Notice that the bundling graph (\cref{def:bundlinggraph}) contains a cycle corresponding to each variable, as depicted in \cref{fig:gadget b=2}.  

  The general idea is that in order to let $p$ win, all \voterstype{p}
  must be in $\combRule(\electionW')$ and no clause alternative~$c_i$ should
  gain more than $(|C_i|-1)$ points.  More formally, we show now that
  $(\mathcal{C}, \mathcal{X})$ has a satisfying truth assignment if
  and only if there is a size-$k$ subset~$W'\subseteq W$ such that $p$
  wins election~$(\electionC, \electionV\cup\combRule(W'))$ (recall
  that $k = 4|\mathcal{X}|$).

  For the ``if'' direction, let $\beta: \mathcal{X} \to \{T, F\}$ be a
  satisfying truth assignment function for~$(\mathcal{C},
  \mathcal{X})$.  Intuitively, $\beta$ will guide us through
  constructing the set $W'$ in the following way: First, for each
  variable $x_j$, we put into $W'$ those voters $c_i^{\ell_j}$ for
  whom $\beta$ sets $\ell_j$ to false. This way in $\combRule(W')$ we
  include $2|\mathcal{X}|$ \voterstype{p} and, for each clause $c_i$,
  at most $(|C_i|-1)$ \voterstype{c_i}. The former is true because exactly
  $|\mathcal{X}|$ literals are set to false by $\beta$, each literal
  is included in exactly two clauses, and adding each $c_i^{\ell_j}$ into
  $W'$ also includes a unique \votertype{p} into $\combRule(W')$; the
  latter is true because if $\beta$ is a satisfying truth assignment
  then each clause~$C_i$ contains at most $(|C_i|-1)$ literals set to false. 
  Then,
  for each clause voter $c_{i}^{\ell_j}$ already in $W'$, we also add the
  voter $p^j_a$, $1 \leq a \leq 4$, that contains $c_{i}^{\ell_j}$ in
  his or her bundle.
  This way we include in $\combRule(W')$ additional
  $2|\mathcal{X}|$ \voterstype{p} without increasing the number of
  clause voters included.  
  Formally, we define $W'$ as follows:
  \begin{align*}
  W':= &\{c_i^{\lnot x_j}, p^j_a \mid \lnot{x_j} \in C_i \wedge \beta(x_j) =
  T \wedge c_i^{\lnot x_j} \in \combRule(p^j_a) \} \cup \\
   &\{c_i^{x_j}, p^j_a \mid x_j \in C_i \wedge \beta(x_j) = F \wedge c_i^{\lnot x_j} \in \combRule(p^j_a)\}\text{.}
  \end{align*}
  As per our intuitive argument, one can verify that all
  \voterstype{p} are contained in $\combRule(\electionW')$ and each
  clause alternative~$c_i$ gains at most $(|C_i|-1)$ points.

  For the ``only if'' part, let $W'$ be a subset of voters such that $p$ wins
  election~$(\electionC, \electionV\cup \combRule(W'))$.
  
  First, we make the following observation. Let $x_j$ be some variable
  and consider clauses $C_i$ and $C_s$ where literal $\lnot x_j$ appears, and
  clauses $C_r$ and $C_t$ where literal $x_j$ appears. We claim that we can
  assume that $\combRule(W')$ contains at most two voters among
  $c_i^{\lnot x_j}$, $c_s^{\lnot x_j}$, $c_r^{x_j}$, and $c_t^{x_j}$.
  First, let us assume that $\combRule(W')$ contains
  all of these voters. Since $p$ is a winner of election $(\electionC,
  \electionV\cup \combRule(W'))$, it must be that $\combRule(W')$
  also contains all four \voterstype{p} of the form $p^j_a$, $1 \leq a \leq 4$. 
  This means that $W'$ includes at least four voters from the set:
  \[ Q_j = \{ c_i^{\lnot x_j}, c_s^{\lnot x_j}, c_r^{x_j}, c_t^{x_j}, p^j_1, p^j_2, p^j_3, p^j_4\}. \]
  In effect, we can replace $W'$ with $W''$ defined as 
  \[
  W'' := (W' \setminus Q_j) \cup \{ c^{x_j}_r\} \cup \{p^j_a \mid c^{x_j}_r \in \combRule(p^j_a)\} \cup
                                \{ c^{x_j}_t\} \cup \{p^j_a \mid c^{x_j}_t \in \combRule(p^j_a)\}.
  \]
  Compared to $W'$, $W''$ contains at most as many voters as $W'$
  does, $\combRule(W'')$ contains the same number of \voterstype{p} as
  $\combRule(W')$ does, and for each clause alternative~$c$,
  $\combRule(W'')$ contains no more \voterstype{c} than
  $\combRule(W')$ does. Thus, $p$ is still a winner of election
  $(\electionC, \electionV\cup \combRule(W''))$ and $W''$ is a valid
  solution.
 
  Furthermore, let us assume that exactly three voters among
  $c_i^{\lnot x_j}$, $c_s^{\lnot x_j}$, $c_r^{x_j}$, and $c_t^{x_j}$
  are included in $\combRule(W')$. For the sake of concreteness, let
  $c_i^{\lnot x_j}$ be the voter not in $\combRule(W')$.
  We use a similar argument as before.
  Specifically, since $p$ is a winner of $(\electionC,
  \electionV\cup \combRule(W'))$, $W'$ must include at least four
  voters among those in $Q_j$. Replacing $W'$ with $W''$ (defined in
  the previous paragraph) works again.
  Notice that, replacing $W'$ with $W''$
  would also work if $c_s^{\lnot x_j}$ was the voter not included in
  $W'$; if either $c_r^{x_j}$ or $c_t^{x_j}$ were the not-included
  voter, we would replace $W'$ with 
  \[W''' := (W' \setminus Q_j) \cup \{
  c^{\lnot x_j}_i\} \cup \{p^j_a \mid c^{\lnot x_j}_i \in
  \combRule(p^j_a)\} \cup \{ c^{\lnot x_j}_s\} \cup \{p^j_a \mid
  c^{\lnot x_j}_s \in \combRule(p^j_a)\}.\]
  

  We will now argue that for each variable $x_j$, $\combRule(W')$
  contains either the two voters of the form $c^{x_j}$ or the two
  voters of the form $c^{\lnot x_j}$. We start by observing that for
  each two clauses that contain the same variable but not the same
  literal, at least one corresponding clause voter must be added to
  the election (otherwise $\combRule(W')$ would not contain all the
  unregistered \voterstype{p}).  Thus, if one clause voter is not
  contained in $\combRule(W')$, then both of its ``neighboring''
  (in the sense of being adjacent in the \bundlinggraph, 
  depicted in \cref{fig:gadget b=2})
  clause voters must be included in $\combRule(W')$.
  Together with the arguments
  from previous paragraphs, this means that for each variable $x_j$,
  $\combRule(W')$ either contains the two voters of the form $c^{x_j}$ or the two
  voters of the form~$c^{\lnot x_j}$.

  This is critical for the sanity of the truth assignment
  function~$\beta$ we will construct now.  In order to let $p$ win,
  all \voterstype{p} must be
  added to the election. This means that for each two clauses that
  contain the
  same variable but not the same literal, at least one corresponding
  clause
  voter must be added to the election.

  We set $\beta: \mathcal{X} \to \{T, F\}$ such that $\beta(x_j) := T$
  if there is a clause voter~$c^{x_j}_{i} \notin \combRule(W')$, and
  $\beta(x_j) := F$ if there is a clause voter~$c^{\lnot x_j}_i \notin
  \combRule(W')$.  Following the previous arguments, function $\beta$ is
  well-defined. It is a satisfying truth assignment function for
  $(C,\mathcal{X})$ because for each clause $C_i$, by the fact that
  $p$ is a winner in election $(\electionC, \electionV\cup
  \combRule(W'))$, we have that $\combRule(W')$ contains at most $(|C_i|-1)$
  \voterstype{c_i} for each clause alternative~$c_i$. 
  This is possible only if each clause contains at
  least one literal $\ell$ such that $\beta$ sets $\ell$ to truth.
\end{proof}

The situation is different for \fullassignments,
because we can extend the greedy algorithm by \citet{BTT92} to bundles of size two.

\begin{theorem}\label{obs:b2fullp}
  \obsbtwofullp
\end{theorem}

\begin{proof}
  Since $b = 2$ and $\combRule$ is a \fullassignment,
  the \bundlinggraph has maximum degree one.
  Therefore, it contains only isolated vertices and disjoint edges.
  We first add the disjoint edges with both end-points corresponding to \voterstype{p}.
  If we have some more budget,
  then we add isolated vertices corresponding to \voterstype{p}.
  We are left only with isolated vertices corresponding to non-\voterstype{p},
  which we throw away,
  disjoint edges with both end-points corresponding to non-\voterstype{p},
  which we also throw away,
  and disjoint edges with one end-point corresponding to a \votertype{p} and another end-point corresponding to a non-\votertype{p},
  which we treat now.
  Specifically, we add these disjoint edges with one end-point corresponding to a \votertype{p} and another end-point corresponding to a non-\votertype{p},
  sorted ascendingly by the current score of the non-\votertype{p}.
\end{proof}

However, as soon as $b = 3$, we obtain $\np$-hardness, 
by modifying the reduction used in \cref{thm:np-h-b=2}.

\begin{theorem}\label{thm:fullB3HARD}
  \thmfullbthreehard
\end{theorem}

\begin{proof}
  We use a similar reduction as in the proof of \cref{thm:np-h-b=2},
  with the only difference that 
  we introduce eight \voterstype{p} for each variable instead of four \voterstype{b}.
  We set the \fullassignment~$\combRule$ such that each variable voter's bundle consists of two variable voters and one clause voter,
  and such that each clause voter's bundle also consists of two variable voters and one clause voter.
  Now the cycle corresponding to each variable consists of twelve vertices, as depicted in \cref{fig:gadget b=3full}.
  Moreover, $\combRule$ is full-$d$ for some $d$.
  The correctness proof is analogous to the one shown for \cref{thm:np-h-b=2}.
\end{proof}

Taking also the \swapdist $d$ into account,
we find out that both
\textsc{Plurality}-\probCCCAV and \textsc{Condorcet}-\probCCCAV
are $\np$-hard, even if $d = 1$.
This stands in contrast to the case where $d = 0$,
where $\calR$-\probCCCAV reduces to the \probCCAV problem
(perhaps for the weighted voters~\cite{FHH13}),
which,
for Plurality voting,
is polynomial-time solvable by a simple greedy algorithm.

\begin{theorem}\label{prop:np-hard:d=1:b=4}
  \thmnphdonebfour
\end{theorem}

\begin{proof}
   The theorem follows from the proof of
   Theorem~\ref{thm:SinglepeakedWhardness}, applied to a reduction from
   \probVC, for graphs with maximum vertex degree equal to three.
   \probVC remains $\np$-complete in this case~\cite{GJS76}.
\end{proof}

\section{Single-Peaked and Single-Crossing Elections}
\label{sec:sp+sc}
In this section, we focus on instances with \fullassignments,
and we do so because without this restriction the hardness results from
previous sections easily translate to our restricted domains (at least
for the case of the Plurality rule). We find that the results for the
combinatorial variant of control by adding voters for single-peaked
and single-crossing elections are quite different than those for the
non-combinatorial case. Indeed, both for Plurality and for Condorcet,
the voter control 
problems for single-peaked
elections and for single-crossing elections are solvable in polynomial
time for the non-combinatorial case~\cite{BBHH2010,FHHR11,MF14-ecai}.
For the combinatorial case,
we show hardness for both \textsc{Plurality}-\probCCCAV and \textsc{Condorcet}-\probCCCAV
for single-peaked elections,
but give polynomial-time algorithms for single-crossing elections.
We mention that the intractability results can also be seen as regarding 
anonymous \assignments because all \fullassignments are \leaderanonymous and \followeranonymous.

We begin with single-peaked elections.

\begin{theorem}\label{thm:SinglepeakedWhardness}
 \thmspwonehk
\end{theorem}

\begin{proof}
  \newcommand{\winitialscore}{\ensuremath{h+\ell}}
  \newcommand{\order}[1]{\mathrm{diff}\text{-}\mathrm{order}(#1)}
  We provide a parameterized reduction from the $\wone$-hard problem~\probPVC~(\PVC) with respect to the ``solution size'' parameter~$h$~\cite{GNW07},
  which asks for a set of at most $h$ vertices in a graph $G$, which intersects with at least $\ell$ edges.
  More formally:

   \probDef
     {\probPVC (\PVC)}
     {An undirected graph $G=(V(G), E(G))$ and $h, \ell \in \mathbb{N}$.}
     {Does $G$ admits a size-$h$ vertex subset~$U \subseteq V(G)$ which intersects at least $\ell$ edges in $G$?}

  Given a \PVC instance~$(G, h, \ell)$,
  we set $k:=h$,
  and construct an election~$E=(\electionC, \electionV)$ with $C:=\{p,w\}\cup \{a_i, \overline{a}_i, b_i, \overline{b}_i \mid u_i \in V(G)\}$
  such that the initial score of $w$ is $\winitialscore$
  and the initial scores of all the other alternatives are zero.
  We do so by creating $\winitialscore$ registered voters who all have the same preference order~$\pref$
  such that it differs from the following
  \emph{canonical preference order}:
  \[p \pref w \pref a_1 \pref \overline{a}_1 \pref \ldots \pref a_{|V(G)|} \pref \overline{a}_{|V(G)|} \pref b_1 \pref \overline{b}_1 \pref \ldots \pref b_{|V(G)|} \pref \overline{b}_{|V(G)|}
  \]
  by only the first pair~$\{p, w\}$ of alternatives.

  For each set $P$ of disjoint pairs of alternatives,
  neighboring with respect to the canonical preference order,
  we define the preference order $\order{P}$ to be identical to the canonical preference order,
  except that all the pairs of alternatives in $P$ are swapped.
  The unregistered voter set~$\electionW$ is constructed as follows:
  \begin{enumerate}[(i)]
  \item for each edge $e = \{u_i,u_j\} \in E(G)$, we create an \emph{edge
      voter}~$w_e$ with preference order~$\order{\{\{a_i,
      \overline{a}_i\}, \{a_j, \overline{a}_j\}\}}$ (we say that $w_e$
    corresponds to edge~$e$),
  \item for each edge $e = \{u_i,u_j\} \in E(G)$, we create a \emph{dummy
      voter}~$d_e$ with preference order~$\order{\{\{p,w\},\{a_i,
      \overline{a}_i\}, \{a_j,$ $\overline{a}_j\}\}}$ (we say that
    $d_e$ corresponds to edge~$e$), and
  \item for each vertex $u_i \in V(G)$, we create a \emph{vertex voter}~$w^{u}_{i}$ with preference
    order~$\order{\{\{a_i, \overline{a}_i\}\}}$ (we say that $w^{u}_i$ corresponds to $u_i$). 

  \end{enumerate} 
  The preference orders of the voters in $\electionV \cup \electionW$ are single-peaked with respect to the axis
  \[\seq{\overline{B}} \pref \seq{\overline{A}} \pref p \pref w \pref \seq{A} \pref \seq{B},\]
  where
  \begin{alignat*}{2}
  \seq{\overline{B}} &:= \overline{b}_{|V(G)|} \pref \overline{b}_{|V(G)|-1} \pref \ldots \pref \overline{b}_{1}, &\quad \seq{\overline{A}} &:= \overline{a}_{|V(G)|} \pref \overline{a}_{|V(G)|-1} \pref \ldots \pref \overline{a}_{1},\\
  \seq{B} &:= b_1 \pref b_2 \pref \ldots \pref b_{|V(G)|}, &\quad
  \seq{A} &:= a_1 \pref a_2 \pref \ldots \pref a_{|V(G)|}. 
\end{alignat*}
  
Finally, we define the function~$\combRule$ such that it is a \fullassignment[1].
  To understand how $\combRule$ works, we carefully
  calculate the swap distance between the preference orders of all
  possible pairs of voters in $W$. We see that:
  \begin{enumerate}[(a)]
  \item any two edge voters have swap distance at least two,
  \item any edge voter and any dummy voter have swap distance exactly
    one if they correspond to the same edge, and at least three
    otherwise,
  \item any edge voter~$w_e$ and any vertex voter~$w^{u}_i$ have swap
    distance one if $u_i \in e$, and three otherwise,
  \item any two dummy voters have swap distance at least two,
  \item any dummy voter and any vertex voter have swap distance at
    least two, and
  \item any two vertex voters have swap distance two.
  \end{enumerate}
  Thus, for each edge $e = \{u_i, u_j\} \in E(G)$ we have
  $\combRule(w_e) := \{w_e, w^{u}_i, w^{u}_j, d_e\}$ and
  $\combRule(d_e) := \{w_e, d_e\}$, and for each vertex $u_i \in V(G)$
  we have $\combRule(w^u_i) := \{w^{u}_i\} \cup \{w_e \mid u_i \in e
  \in E(G)\}$.

  We show that $(G, h, \ell)$ is a yes-instance for \PVC if and only
  if there is a size-$k$ subset~$W'\subseteq W$ such that $p$ is a
  Plurality winner of the election~$(\electionC, \electionV\cup
  \combRule(\electionW'))$.  Note that all unregistered voters except
  the dummy voters prefer~$p$ over all other alternatives and that $p$
  needs at least $\winitialscore$~points in order to win.
  
  For the ``only if'' part, suppose that $X\subseteq V(G)$ is a
  size-$h$ vertex set and $Y\subseteq E(G)$ is a size-$\ell$ edge set
  such that for every edge $e\in Y$ it holds that $e\cap X\neq
  \emptyset$.  We set $\electionW':=\{w^{u}_i \mid u_i \in X\}$, and
  it is easy to verify that $\combRule(\electionW')$ consists of $h$
  vertex voters and at least $\ell$ edge voters.  Each of them gives
  $p$ one point if added to the election.  This results in $p$ being a
  winner of the election with score at least $\winitialscore$.

  For the ``if'' part, suppose that there is a size-$k$
  subset~$W'\subseteq W$ such that $p$ is a Plurality winner of
  the election~$(\electionC, \electionV\cup \combRule(\electionW'))$.
  Observe that if $\electionW'$ contains some dummy voter~$d_e$, then
  we can replace it with $w_e$ (if $w_e$ is already in
  $\electionW'$ then we can simply remove $d_e$ from $\electionW'$).
  Thus we can assume that $\electionW'$ does not contain any dummy
  voters. Now, assume that $\electionW'$ contains some edge
  voter~$w_e$, where $e = \{u_i, u_j\}$. Since, by the previous argument,
  $\electionW'$ does not contain~$d_e$, we have that $d_e$ is not a
  member of $\combRule(\electionW' \setminus \{w_e\})$. This means
  that if both $w_{u_i}$ and $w_{u_j}$ belong to
  $\combRule(\electionW' \setminus \{w_e\})$ then we can safely remove
  $w_e$ from $W'$; $p$ will still be a winner of the election~$(C,\electionV \cup \combRule(W' \setminus \{w_e\}))$. On the other hand,
  assume that exactly one of $w_{u_i}$, $w_{u_j}$ does not belong to
  $\combRule(\electionW' \setminus \{w_e\})$ and let $w_u$ be this
  voter.  It is easy to see that $p$ is a winner of election
  $(\electionC, \electionV\cup \combRule((\electionW' \setminus
  \{w_e\}) \cup \{w_u\}))$ (the net effect of including the bundle of
  $w_e$ is that $p$'s score increases by at most one, whereas the net
  effect of including the bundle of $w_u$ is that $p$'s score
  increases by at least one). Similarly, if neither $u_i$ nor $u_j$
  belong to $\combRule(\electionW' \setminus \{w_e\})$, then it is easy
  to verify that $p$ is a winner of the election~$(\electionC,
  \electionV\cup \combRule((\electionW' \setminus \{w_e\}) \cup
  \{w_{u_i}\}))$. All in all, we can assume that $\electionW'$ contains
  vertex voters only. Since all vertex voters are $p$-voters, without
  loss of generality we can assume that $\electionW'$ contains exactly $k = h$ of
  them.
  

  We define $X:=\{u_i \mid w^{u}_i \in \electionW'\}$ such that $|X| = k$, and $Y:=\{e \in E(G) \mid e \cap u_i \neq \emptyset\}$.
  By the construction of the edge voters' preference orders, 
  $\combRule(\electionW')$ consists of $k$ vertex voters and $|Y|$ edge voters.
  This must add up to at least $\winitialscore$ voters.
  Therefore, $|Y|\ge \ell$, implying that at least $\ell$ edges are covered by $X$.\medskip

  As for the Condorcet rule, 
  we use the same unregistered voters as defined above 
  and construct the original election with
  $\winitialscore-1$ registered voters whose preference orders are $\order{\{w,p\}}$.
  Using the same reasoning as used for the Plurality rule, 
  one can verify that $(G,h,\ell)$ is a yes-instance for \PVC if and only if there is a size-$k$ subset $W'\subseteq W$ such that $p$ is a Condorcet winner of the election~$(\electionC, \electionV\cup \combRule(\electionW'))$.
\end{proof}

We now present some tractability results for single-crossing elections.
Consider an $\calR$-$\probCCCAV$ instance~$\probCCCAVInstance$,
containing an election $(C,V)$ and an unregistered voter set~$W$
such that $(C,V \cup W)$ is single-crossing,
 and thus, both $(C,V)$ and $(C,W)$ are single-crossing.
This has a crucial consequence for
\fullassignments: For each unregistered voter~$w\in
\electionW$, the voters in bundle~$\combRule(w)$ appear consecutively
along the single-crossing order restricted to only the voters in
$\electionW$.\footnote{Note that for each single-crossing election, the
  order of the voters possessing the single-crossing property is, in essence,
  unique. (modulo voters with the same preference orders and modulo the
  fact that if an order witnesses the single-crossing property of an election,
  then its reverse does so as well).}
Using the following lemmas, we can show that \textsc{Plurality}-\probCCCAV and
\textsc{Condorcet}-\probCCCAV are polynomial-time solvable in some cases.


\begin{lemma}\label{lem:plurality-single-crossing}
 \lemplusc
\end{lemma}

\begin{proof}
  Let $n:= |\electionW|$ and 
  let $\alpha:=\seq{w_1,w_2,\ldots,w_{n}}$ be a single-crossing order of the voters in $\electionW$.
  Item~(\ref{lem:p-voters-consecutive}) follows directly from the definition of the single-crossing property. 
  
  As for Item~(\ref{lem:|non-p-voter-bundles|<=2}), 
  let $\electionW'\subseteq \electionW$ be a size-$k$ subset of unregistered voters 
  such that $p$ is a Plurality winner in election~$(\electionC, \electionV \cup \combRule(\electionW'))$.
  For each subset~$S\subseteq W$ of voters,
  we use $\firstvoter{S}$ (resp.\ $\secondvoter{S}$)
  to denote the index~$j$ (resp. $j'$) of the first voter~$w_j\in S$ (resp.\ the last voter~$w_{j'} \in S$) along the single-crossing order. 
  Suppose that there are two bundles, $\combRule(w_i)$ and $\combRule(w_j)$, with $\firstvoter{\combRule(w_i)} \le \firstvoter{\combRule(w_j)}$
  such that both contain non-\voterstype{p} and the first \votertype{p} along $\alpha$.
  If $\secondvoter{\combRule(w_i)} \le \secondvoter{\combRule(w_j)}$,
  then $\combRule(w_i)$ does not contain more \voterstype{p} than $\combRule(w_j)$ does,
  while containing at least as many non-\voterstype{p} as $\combRule(w_j)$.
  Thus, we can remove $w_i$ from $\electionW'$.
  Otherwise, $\secondvoter{\combRule(w_i)} > \secondvoter{\combRule(w_j)}$,
  which means that $\combRule(w_j) \subset \combRule(w_i)$.
  Thus, we can remove $w_j$ from $\electionW'$.
  In any case, we conclude that $\electionW'$ contains at most one voter~$w$ whose bundle~$\combRule(w)$ contains a non-\votertype{p} and the first \votertype{p} (along the single-crossing order).
  
  Analogously, we can show that $\electionW'$ contains at most one voter~$w$ whose bundle~$\combRule(w)$ contains a non-\votertype{p} and the last \votertype{p} (along the single-crossing order).
  Since for each bundle~$\combRule(w)$ with $w\in \electionW'$, 
  if $\combRule(w)$ contains a non-\votertype{p}, 
  then it contains at least one of the first and last voters along~$\alpha$,
  every bundle $\combRule(w)$ with $w\in \electionW'$ contains at least one \votertype{p}
  (because if it does not, then we can remove its respective leader voter, as the bundle does not help $p$),
  and Item~(\ref{lem:|non-p-voter-bundles|<=2}) follows.
\end{proof}

 For Condorcet voting, we use the well-known median-voter theorem (we
 provide the proof for the sake of completeness).

\begin{lemma}\label{lem:condorcet-single-crossing}
  \lemcondsc
\end{lemma}

\begin{proof}
  Let $X_1$ be the set of voters~$x_1,x_2,\ldots, x_{\lceil z/2 \rceil - 1}$ and
  let $X_2$ be the set of voters~$\electionV \cup \combRule(\electionW') \setminus (X_1 \cup X_{\mathrm{median}})$.

  For the ``if'' part, let $c$ be an arbitrary alternative from $\electionC \setminus \{p\}$.
  Then, if there is some voter in $X_1$ which prefers $c$ over $p$,
  then all voters in $X_{\mathrm{median}}\cup X_2$ prefer $p$ over $c$.
  If there is some voter in $X_2$ which prefers $c$ over $p$,
  then all voters in $X_1 \cup X_{\mathrm{median}}$ prefer $p$ over $c$.
  In any case, a strict majority of voters prefer $p$ over $c$.
  Thus, $p$ is the (unique) Condorcet winner.

  For the ``only if'' part, suppose for the sake of contradiction that 
  there is a voter in $X_{\mathrm{median}}$ which is not a \votertype{p} but a \votertype{c} with $c\in \electionC \setminus \{p\}$. 
  Then, analogously to the reasoning above, 
  at least half of the voters will prefer $c$ over $p$---a contradiction.  
\end{proof}

 With these two lemmas available,
 we give polynomial-time
 algorithms for both \textsc{Plurality}-\probCCCAV and
 \textsc{Condorcet}-\probCCCAV, for the case of single-crossing
 elections and \fullassignments. 


\begin{theorem}\label{thm:cccav-P-sc}
  \thmPsc
\end{theorem}

\begin{proof}
  \newcommand{\diff}[2]{\ensuremath{\combRule(w_{#1})\setminus\combRule(w_{#2})}}
    \newcommand{\probMaxIntCov}{\textsc{Maximum Interval Cover}\xspace}
  First, we find a (unique) single-crossing voter order for $(\electionC, \electionV \cup \electionW)$ in quadratic time~\cite{EFS12,BCW13}.
  Due to~\cref{lem:plurality-single-crossing} and~\cref{lem:condorcet-single-crossing},
  we only need to store the most preferred alternative of each voter
  to find the solution set~$W'$.
  Thus, the running-time from now on only depends on the number of voters.
  We start with the Plurality rule
  and let $\alpha:=\seq{w_1,w_2,$ $\ldots,w_{|W|}}$ be a single-crossing voter order.
  

  Due to
  \cref{lem:plurality-single-crossing}~(\ref{lem:|non-p-voter-bundles|<=2}),
  the two bundles in $\combRule(\electionW')$ which may contain
  non-\voterstype{p} appear at the beginning and at the end of the
  \votertype{p} block, along the single-crossing order. 
  We first guess these two bundles,
  and after this initial guess, all remaining bundles in the solution contain only
  \voterstype{p}~(\cref{lem:plurality-single-crossing}~(\ref{lem:p-voters-consecutive})).
  Thus, the remaining task is to find the maximum score that $p$ can
  gain by selecting $k'$~bundles containing only \voterstype{p}.
  This problem is equivalent to the \probMaxIntCov problem,
  which is solvable in $O(|W|^2)$ time (\citet[Section 3.2]{GKKSS09}).
  
      

  For the Condorcet rule,
  we propose a slightly different algorithm.
  The goal is to find a minimum-size subset~$\electionW'\subseteq \electionW$ such that
  $p$ is the (unique) Condorcet winner in $(\electionC, V\cup \combRule(\electionW'))$.
  Let $\beta:=\seq{x_1, x_2, \ldots, x_{z}}$ be a single-crossing voter order for $(\electionC, \electionV \cup \electionW)$.
  Considering \cref{lem:condorcet-single-crossing},
  we begin by guessing at most two voters in $\electionV\cup \electionW$
  whose bundles may contain the median \votertype{p} (or, possibly, several \voterstype{p})
  along the single-crossing order of voters restricted to the final election
  (for simplicity, we define the bundle of each registered voter to be its singleton).
  The voters in the union of these two bundles must be consecutively ordered.
  Let those voters be~$x_i, x_{i+1}, \ldots, x_{i+j}$ (where $i\ge 1$ and $j\ge 0$),
  let $\electionW_1 := \{x_s \in \electionW \mid s < i\}$,
  and let $\electionW_2 := \{x_s \in \electionW \mid s>i+j\}$. 
  We guess two integers~$z_1 \le |\electionW_1|$ 
  and $z_2 \le |\electionW_1|$ 
  with the property that
  there are two subsets~$B_1\subseteq \electionW_1$ 
  and $B_2\subseteq \electionW_2$ with $|B_1|=z_1$ and $|B_2|=z_2$ such that
  the median voter(s) in $V\cup B_1 \cup \{x_i, x_{i+1}, \ldots, x_{i+j}\} \cup B_2$ are indeed \voterstype{p}
  (for now, only the sizes $z_1$ and $z_2$ matter, not the actual sets). 
  These four guesses cost $O(|\electionV \cup \electionW|^2 \cdot |\electionW|^2)$ time.
  The remaining task is to find two minimum-size subsets~$\electionW'_1$
  and $\electionW'_2$ such that $\combRule(\electionW'_1) \subseteq \electionW_1$, $\combRule(\electionW'_2)\subseteq W_2$, $|\combRule(\electionW'_1)|=z_1$, and $|\combRule(\electionW'_2)|=z_2$.
  As already discussed, this can be done in $O(|W|^2)$ time~\cite{GKKSS09}.
  We conclude that
  one can find a minimum-size subset~$\electionW'\subseteq \electionW$ such that
  $p$ is the (unique) Condorcet winner in $(\electionC, V\cup \combRule(\electionW'))$ in $O(|V\cup W|^2\cdot |W|^4)$ time.
\end{proof}



\section{Conclusion}
%
%
%
We provide 
opportunities for future research.  First, we did not
discuss destructive control and the related problem of
combinatorial deletion of voters.
For Plurality, we conjecture that combinatorial addition of voters for
destructive control, and combinatorial deletion of voters for either
constructive or destructive control behave similarly to
combinatorial addition of voters for constructive control.

Another, even wider field of future research is to study other
combinatorial voting models---this may include controlling the swap
distance, ``probabilistic bundling'', ``reverse bundling'', or using
other distance measures than the swap distance. Naturally, it would
also be interesting to consider other problems than election control
(with bribery 
being perhaps the most natural candidate).

Finally, instead of studying a ``leader-follower model'' as we did,
one might also be interested in an ``enemy model'' referring to
control by adding alternatives: The alternatives of an election
``hate'' each other such that if one alternative is added to the
election, then all of its enemies are also added to the election.
This scenario of combinatorial candidate control deserves future
investigation.


\newcommand{\bibremark}[1]{}

\end{document}